\documentclass[11pt]{article}
\usepackage{amsmath,amsthm,amssymb}
\usepackage{verbatim}
\usepackage{anysize}

\usepackage{graphicx}

\usepackage{./algorithm, ./algorithmic}

\marginsize{2.5cm}{2.5cm}{1.5cm}{2.5cm}

\theoremstyle{plain}
\newtheorem{theorem}{Theorem}
\newtheorem{lemma}{Lemma}

\theoremstyle{definition}
\newtheorem{definition}{Definition}

\theoremstyle{remark}

\begin{document}
\title{Approximating the Regular Graphic TSP in near Linear Time}
\author{Ashish Chiplunkar \qquad Sundar Vishwanathan\\
Department of Computer Science and Engineering\\
Indian Institute of Technology Bombay\\
Mumbai India\\
\texttt{\{ashishc, sundar\}@cse.iitb.ac.in}
}
\date{}
\maketitle

\begin{abstract}
We present a randomized approximation algorithm for computing traveling salesperson tours in undirected regular graphs. Given an $n$-vertex, $k$-regular graph, the algorithm computes a tour of length at most $\left(1+\frac{7}{\ln k-O(1)}\right)n$, with high probability, in $O(nk \log k)$ time. This improves upon a recent result by Vishnoi (\cite{Vishnoi12}, FOCS 2012) for the same problem, in terms of both approximation factor, and running time. The key ingredient of our algorithm is a technique that uses edge-coloring algorithms  to sample a cycle cover with $O(n/\log k)$ cycles with high probability, in near linear time.

Additionally, we also give a deterministic $\frac{3}{2}+O\left(\frac{1}{\sqrt{k}}\right)$ factor approximation algorithm running in time $O(nk)$.
\end{abstract}

\section{Introduction}

Given a complete undirected graph with  positive real valued weights on the edges, the traveling salesperson problem ({\sc{TSP}}) is to find a minimum weight cycle that visits each vertex exactly once.
This problem was among the first few proved NP-Complete by Karp \cite{Karp72}.
In  the nearly four decades since this result was published, the {\sc{TSP}} has been a fundamental research problem in both complexity theory, and algorithmic graph theory. 
In the absence of any structural restriction on the weight function, the {\sc{TSP}} is hard to approximate within any constant factor (\cite{Trevisan04}, \cite{SahniG76}). 
The most widely researched restriction of the TSP is the {\sc{MetricTSP}}, where the vertices form a metric space with the weight function as the metric. This simple imposition of the triangle inequality over the weights allowed Christofides \cite{Christofides76} to efficiently construct tours with an approximation ratio of $3/2$. No progress has been made on this upper bound in the last 35 years.

The lower bound on the approximation factor has undergone steady improvement. The first explicitly proven lower bound was $5381/5380$ by Engebretsen \cite{Engebretsen03} (for the {\sc{MetricTSP}} with distances 1, and 2), followed by $3813/3812$ by B\"{o}ckenhauer and Seibert \cite{BockenhauerS00} (for the {\sc{MetricTSP}} with distances 1, 2, and 3).
Consequently, Papadimitriou and Vempala \cite{PapadimitriouV06} demonstrated the NP-hardness of approximating within a factor better than $220/219$. This bound was later improved by Lampis \cite{Lampis12} to $185/184$. The current lower bound is $123/122$, given by Karpinski, Lampis, and Schmied \cite{KarpinskiLS13}. We refer the reader to \cite{KarpinskiLS13} for a nice overview of recent advances in many natural restrictions of the {\sc{MetricTSP}}. 

Returning to the upper bound, the existence of a PTAS was precluded early on by the proof of APX-hardness of the {\sc{MetricTSP}} given by Papadimitriou and Yannakakis \cite{PapadimitriouY93}. Arora \cite{Arora98} and Mitchell \cite{Mitchell99} constructed a PTAS for the case when the metric is Euclidean (the {\sc{EuclideanTSP}}). Interestingly, Christofides \cite{Christofides76} proposed his algorithm for the {\sc{EuclideanTSP}}, however, his proof only relied on the triangle inequality property of the Euclidean metric, making it valid for all metrics. 
The existence of a PTAS clearly separates the {\sc{EuclideanTSP}} from the {\sc{MetricTSP}} in terms of computational hardness.

An important sub-class of the {\sc{MetricTSP}} is the {\sc{GraphTSP}}, where the weight function on the edges arises from the shortest path distances in some unweighted undirected graph. This is believed to be the most promising candidate for capturing the computational hardness of the {\sc{MetricTSP}}. {\sc{GraphTSP}} is APX-hard (as a consequence of its MAX-SNP hardness \cite{PapadimitriouY93} and the PCP theorem \cite{AroraLMSS98}), and the best known lower bound of $4/3$ on the integrality gap of the Held-Karp LP relaxation \cite{HeldKarp71} of the {\sc{MetricTSP}} is observed on an instance of the {\sc{GraphTSP}}. Gharan, Saberi and Singh \cite{GharanSS11} achieved the first improvement over Christofides \cite{Christofides76} algorithm for the {\sc{GraphTSP}} with an approximation ratio strictly less then $3/2$, which was shortly followed by M\"{o}mke and Svensson's \cite{MomkeS11} bound of $1.461$. Mucha \cite{Mucha12} later improved the analysis of M\"{o}mke and Svensson's \cite{MomkeS11} algorithm and demonstrated a bound of $13/9$. Currently, the best known bound is $7/5$, given by Seb\"{o} and Vygen \cite{SeboV12}. It is widely believed that the Held-Karp relaxation has an integrality gap of precisely $4/3$, and this has been proven for cubic graphs \cite{BoydSSS14}.

Vishnoi \cite{Vishnoi12} opened up a new line of interesting work by arguing that the Held-Karp relaxations of the {\sc{GraphTSP}} possibly get better with increasing edge density. Alluding to Dirac's theorem \cite{Dirac52} assuring existence of a Hamiltonian cycle in graphs with min-degree greater than $n/2$, Karp's result \cite{Karp76} on random graph models containing efficiently constructible Hamiltonian cycles, and his own treatment of the TSP on k-regular graphs ({\sc{RegGraphTSP}}), Vishnoi puts forth a view, that the hardness of the {\sc{GraphTSP}} might indeed be concentrated among instances of low-degree. We refer the reader to Vishnoi \cite{Vishnoi12} for a nice survey on the {\sc{MetricTSP}} in general, and an interesting discussion on this line of work.

The main technical contribution of Vishnoi's paper is a novel method which efficiently constructs arbitrarily good tours for the {\sc{RegGraphTSP}}. The algorithm first constructs a cycle cover with $O(n/\sqrt{\ln k})$ cycles, using Jerrum, Sinclair and Vigoda's algorithm \cite{JerrumSV04} for sampling a perfect matching from the bipartite version of the input graph. These cycles are then connected using two copies of a spanning tree on the graph formed by contracting the cycles. This yields a  tour of length at most  $(1+\sqrt{64/\ln k})n$ with probability $1-1/n$. The running time of this algorithm is dictated by the sampling method, which is around $O(n^{10}\log^3n)$.  This can be improved marginally by using a faster sampling algorithm, for example, the algorithm by Bez{\'a}kov{\'a}, Stefankovic, Vazirani and Vigoda \cite{BezakovaSVV08}. In this paper, we propose an alternative method, which solves the {\sc{RegGraphTSP}} with a better approximation factor, and in near linear time.

\begin{theorem}\label{thm_main}
There is an algorithm which, given a connected $k$-regular undirected graph on $n$ vertices, runs in time $O(nk\log k)$, and outputs a TSP tour of cost at most $\left(1+\frac{7}{\ln k-O(1)}\right)n$ with high probability (specifically, probability of failure decaying exponentially with $n$).
\end{theorem}

The inspiration behind our work, apart from Vishnoi's result, is the large body of work on fast 
algorithms for finding matchings and edge-colorings in regular bipartite graphs. The seminal
work in this area is an ingeniously simple algorithm due to Gabow and Kariv \cite{GabowK82}. 
Noteworthy research on this topic include
those by Schrijver \cite{Schrijver98}, Alon \cite{Alon03}, 
and Cole, Ost and Schirra \cite{ColeOS01}. The last gives the best known
bound for coloring bipartite multi-graphs with $m$ edges and maximum degree $D$ in 
time $O(m \log D)$.

Our search for a faster algorithm begins with the following question.
Given a fast algorithm for finding a  matching in regular bipartite graphs,
can one introduce randomization and produce an almost uniform distribution
on the matchings output? 
The essential problem with this approach is understanding and dealing with the inherent sampling biases towards various matchings.
One of our ideas towards mitigating this effect is obtaining a coloring at random and considering 
every matching in the coloring.  It helps that most of the simple algorithms
for matching also find a coloring with roughly similar running times.
It seems prudent to consider the simplest of the
above algorithms and this, without question, is the one by Gabow and Kariv. Unfortunately even this
seems difficult, as we shall discuss in the coming sections. However, we show
that a modification of the Gabow-Kariv algorithm does make it amenable to such an analysis.

Our algorithm improves upon Vishnoi's \cite{Vishnoi12} method by using a faster sampling procedure.
The corresponding sampling distribution may be quite far from uniform.
However, we demonstrate bounds on the measure concentration around cycle covers with few cycles, using simple counting arguments. The algorithm will be described in Section \ref{sec_main_alg}, followed by its analysis in Section \ref{sec_main_analysis}.

Derandomizing our algorithm seems like a difficult problem. We initiate research on deterministic linear time algorithms for the {\sc{RegGraphTSP}} by giving a simple algorithm that achieves a $\frac{3}{2}+O\left(\frac{1}{\sqrt{k}}\right)$ factor approximation. 
Here, the main idea is to traverse the graph in a depth-first-like manner and keep removing long cycles. These cycles cover a good fraction of the vertices. The cycles and the uncovered vertices can then be connected by a spanning tree. We devote Section \ref{sec_det_alg} for this algorithm and its analysis.

\section{The Randomized Algorithm}\label{sec_main_alg}

The high level idea behind our algorithm is similar to that of Vishnoi's. Find a cycle cover of the graph, and then connect the cycles using a spanning tree. Recall that a cycle cover of a graph is a collection of vertex-disjoint cycles that cover all its vertices. We wish to construct a cycle cover such that it has a small number of cycles with high probability. It is folklore that cycle covers in a graph correspond to matchings in the natural encoding of the given graph as a bipartite graph (see Lemma \ref{lem_deg_red}). Indeed, Vishnoi selects a random matching in such an encoding.

Given a $k$-regular graph, we intend to first partition the edges into $k$ cycle covers, and then select the best one. Our algorithm to find the partition uses ideas from the Gabow-Kariv algorithm, which finds a minimum edge-coloring of an input graph. Indeed, partitioning edges of a graph into cycle covers amounts to finding an edge-coloring of its bipartite encoding. However, the Gabow-Kariv algorithm works only on graphs with vertex degrees which are powers of two. This is our first challenge. A regular graph (especially of odd degree) need not have regular subgraphs of all degrees. To get over this we will work with regular directed graphs.

\begin{definition}
We say that a directed graph is $k$-regular if the in-degree as well as the out-degree of each vertex is $k$.
\end{definition}

The reason for working with directed graphs is that one can effectively partition the edges of a $k$-regular directed graph into $k$ cycle covers. As a consequence we have the following lemma, which ensures there is no loss of generality if we restrict our attention to the case where the degree $k$ is a power of two. This lemma relies on the algorithm by Cole, Ost, and Schirra \cite{ColeOS01}, which partitions the edges of any given $k$-regular bipartite undirected graph with $n$ vertices into perfect matchings, and runs in time $O(nk\log k)$.

\begin{lemma}\label{lem_deg_red}
Given a $k$-regular directed graph $G'=(V,A)$ with $n$ vertices and $k'<k$, there is an algorithm which outputs a $k'$-regular subgraph of $G'$, and runs in time $O(nk\log k)$.
\end{lemma}

\begin{proof}
The algorithm constructs an undirected bipartite graph $B=(V_L,V_R,F)$ as follows. For each $v\in V$ it puts vertices $v_L\in V_L$ and $v_R\in V_R$. For each arc $(u,v)\in A$, it puts an edge $(u_L,v_R)$ in $F$. Clearly $B$ is a $k$-regular bipartite graph, and the cycle covers in $G'$ are in one-to-one correspondence with the perfect matchings in $B$. The algorithm then partitions the edges of $B$ into perfect matchings, using the Cole-Ost-Schirra algorithm, and then deletes an arbitrary set of $k-k'$ matchings. This gives a $k'$-regular bipartite subgraph $B'$ of $B$. The algorithm returns the subset of $A$ consisting of arcs which survive in $B'$. (Note that the resultant graph need not be connected.)
\end{proof}


Henceforth, we will assume that $k$ is a power of $2$. Otherwise, if $2^l<k<2^{l+1}$ for some $l$, we will use the algorithm from Lemma \ref{lem_deg_red} to obtain a $2^l$-regular subgraph of the given graph, and use this subgraph for finding the cycle cover. As stated earlier, our algorithm to find the cycle cover is inspired from the Gabow-Kariv algorithm \cite{GabowK82}.
The Gabow-Kariv algorithm uses the classic Divide-and-Conquer paradigm. It first finds an Euler tour of the input graph. The divide step comprises of partitioning the edge set into two, putting alternate edges in the tour in different parts. This is followed by recursing on the two parts. The natural place to introduce randomness here is in the construction of the Euler tour, but we are not able to analyze this. We facilitate analysis by keeping the divide step simple, and pushing the work to the conquer part. Details follow.


\begin{definition}
Let $G'=(V,A)$ be a $k$-regular directed graph. A \textit{cycle cover coloring} of this graph is a partition of the arc set into $k$ cycle covers. Formally, it is a function $c:A\longrightarrow\{1,\ldots,k\}$, such that for each $i\in\{1,\ldots,k\}$, the set $c^{-1}(i)$ is a cycle cover of $G'$. 
\end{definition}

In other words, for any vertex $v$ and color $i$, exactly one arc leaving $v$ and exactly one arc entering $v$ have color $i$. 
It is not difficult to see that the Cole-Ost-Schirra \cite{ColeOS01} result implies that regular directed graphs have efficiently constructible cycle cover colorings. Our procedure, {\sc{RandCycleCoverColoring}}, which outputs a random cycle cover coloring of a regular directed graph, has the following guarantee.

\begin{algorithm}
\caption{{\sc{RandCycleCoverColoring}}($G$)}
\begin{algorithmic}[1]
\STATE \COMMENT{$G$: a $k$-regular $n$ vertex directed graph with $k$ being a power of $2$; returns a random cycle cover coloring of $G$.}
\STATE If $k=1$ \textbf{return} $G$ with each arc colored $1$.
\STATE Convert $G$ into a $k/2$-regular digraph $H=(V',A')$ with $2n$ vertices, by splitting every vertex $v$ into two vertices: $v_{0}$ and $v_{1}$. Distribute the arcs incident on $v$ randomly among $v_0$ and $v_1$, so that each gets half of the incoming and half of the outgoing arcs.
\STATE Obtain an edge coloring $c':A'\longrightarrow\{1,\ldots,k/2\}$ of $H$ by calling {\sc{RandCycleCoverColoring}}$(H)$.
\STATE Fuse the split vertices back to obtain $G$ with the coloring $c'$. For each $i$, the edges colored $i$ constitute a $2$-regular directed graph. Call it $G_i$.
\STATE For each $i\in\{1,\ldots,k/2\}$ and each connected component in $G_i$, find an Euler tour and put the arcs alternately in sets $S$ and $S'$. Pick one of $S$ and $S'$ at random, and recolor that set of arcs to $i+k/2$.
\end{algorithmic}
\end{algorithm}

\begin{theorem}\label{thm_cyclecovercoloring}
Let $G$ be a $k$-regular directed graph with $n$ vertices, where $k$ is a power of $2$. The algorithm {\sc{RandCycleCoverColoring}}, on input $G$, outputs a random cycle cover coloring of $G$, which with high probability contains a cycle cover with at most $3.5n/\ln k$ components. The algorithm runs in time $O(nk\log k)$.
\end{theorem}

The running time $T(n,k)$ of {\sc{RandCycleCoverColoring}} on a $k$-regular graph with $n$ vertices is given by the recurrence $T(n,k)=T(2n,k/2)+O(nk)$. This yields $T(n,k)=O(nk\log k)$. The claim, that the random cycle cover coloring contains a cycle cover with a small number of components with high probability, is deferred to the next section. 

Theorem \ref{thm_main} follows from Theorem \ref{thm_cyclecovercoloring} in the following manner. Given a $K$-regular undirected graph with $n$ vertices, use the degree reduction algorithm from Lemma \ref{lem_deg_red} to get a regular graph $G=(V,E)$ with degree $k=2^{\lfloor\log_2K\rfloor}$. Construct a directed graph $G'=(V,A)$ in the obvious manner. For each edge $\{u,v\}\in E$, include the arcs $(u,v)$ and $(v,u)$ in $A$. Clearly, $G'$ is a $k$-regular directed graph. Now run the procedure {\sc{RandCycleCoverColoring}} to get a cycle cover coloring of $G'$. Choose the best cycle cover from this cycle cover coloring. This cycle cover contains at most $\frac{3.5n}{\ln k}$ cycles, with high probability. 

The rest of the processing is routine. Replacing the arcs constituting the cycle cover with corresponding edges in $G$, contract these components in $G$, and find a spanning tree of the resulting minor. Duplicate the edges of the spanning tree, so that these edges and the edges in the cycle cover form an Eulerian spanning subgraph of $G$. Find an Euler tour in this graph and short-cut it to get a TSP tour of $G$. The cost of this tour is at most $n+2\times\frac{3.5n}{\ln k}=\left(1+\frac{7}{\ln K-O(1)}\right)n$, and this post-processing can be done in time $O(nk)$, that is linear in the size of the graph.


\section{Analysis of {\sc{RandCycleCoverColoring}}}\label{sec_main_analysis}

We first bound from above the probability of getting any fixed cycle cover coloring.

\begin{lemma}\label{lem_prob_ub}
Consider a fixed cycle cover coloring $c$ of the $k$-regular directed graph $G'=(V,A)$, where $k$ is a power of $2$ and $n=|V|$. The probability that the algorithm outputs $c$ is at most $f(n,k)$, where
\[f(n,k)=\left[\frac{k^k}{(k!)^2}\right]^n\times\frac{1}{2^{k-1}}\]
\end{lemma}

\begin{proof}
By induction on $k$. The claim is trivial for $k=1$. Assume now that $k>1$. Consider the coloring $c':A\longrightarrow\{1,\ldots,k/2\}$ given by
\[c'(e)=\left\{
\begin{array}{ll}
c(e) & \text{if }c(e)<k/2\\
c(e)-k/2 & \text{otherwise}
\end{array}
\right.\]
If a run of the algorithm outputs the coloring $c$ then it must obtain the coloring $c'$ at the end of the recursion step. Having obtained $c'$ at the end of the recursion step, the probability of obtaining $c$ is at most $1/2^{k/2}$, since for each $i\in\{1,\ldots,k/2\}$ the probability that arcs having color $i$ get recolored correctly is at most $1/2$.

Next, in order to obtain the coloring $c'$ at the end of the recursion step, it is necessary that for all $v\in V$ and $i\in\{1,\ldots,k/2\}$, the two edges having their tails (resp. heads) at $v$, colored $i$ in $c'$, must separate during the splitting of the vertex $v$. Thus, the probability that the edges having tails (resp. heads) at $v$ get distributed correctly between $v_0$ and $v_1$ is $2^{k/2}/{{k}\choose{k/2}}$. The probability that the vertex $v$ gets split correctly is $\left[2^{k/2}/{{k}\choose{k/2}}\right]^2$. Therefore, the probability that all $n$ vertices get split correctly is $\left[2^{k/2}/{{k}\choose{k/2}}\right]^{2n}$.

Finally, the probability of obtaining $c'$ after the recursive call, given that all vertices split correctly, is at most $f(2n,k/2)$ by induction. Thus we have
\begin{eqnarray*}
\Pr[\text{algorithm outputs }c] & \leq & \left[\frac{2^{k/2}}{{{k}\choose{k/2}}}\right]^{2n}\times f(2n,k/2)\times\frac{1}{2^{k/2}}\\
 & = & \left[\frac{2^{k/2}}{{{k}\choose{k/2}}}\right]^{2n}\times\left[\frac{(k/2)^{k/2}}{((k/2)!)^2}\right]^{2n}\times\frac{1}{2^{\frac{k}{2}-1}}\times \frac{1}{2^{k/2}}\\
 & = & \left[\frac{k^k}{(k!)^2}\right]^n\times\frac{1}{2^{k-1}}=f(n,k)
\end{eqnarray*}
\end{proof}

Using the fact, $\ln(k!)\geq k\ln k-k$, arising from the Stirling's approximation, we have
\begin{equation}\label{eqn_stirling}
f(n,k)=\left[\frac{k^k}{(k!)^2}\right]^n\times\frac{1}{2^{k-1}}\leq\left[\frac{k^k}{(k/e)^{2k}}\right]^n\times\frac{1}{2^{k-1}}=\left(\frac{e^2}{k}\right)^{kn}\times\frac{1}{2^{k-1}}
\end{equation}

We next bound from above the number of cycle covers with exactly $r$ components.

\begin{lemma}\label{lem_ub}
Let $G'=(V,A)$ be a $k$-regular directed graph with $n$ vertices (where $k$ is not necessarily a power of $2$). The number of cycle covers of $G'$ having $r$ cycles is at most ${{n}\choose{r}}k^{n-r}$.
\end{lemma}

\begin{proof}
Number the vertices of $G'$ arbitrarily. Consider a cycle cover $C\subseteq A$ of $G'$ which has $r$ components, and let $(S_1,\ldots,S_r)$ be the partition of $V$ induced by $C$, where $S_1,\ldots,S_r$ are sorted by the the smallest numbered vertices that they contain. We associate the tuple $(|S_1|,\ldots,|S_r|)$ with $C$.

Given a tuple $(s_1,\ldots,s_r)$ such that $\sum_{i=1}^rs_i=n$, let us upper bound the number of cycle covers $C$ of $G'$ that could be associated with this tuple. First note that each $s_i\geq2$ and hence $r\leq n/2$. Let $(S_1,\ldots,S_r)$ be the partition induced by $C$, sorted by the the smallest numbered vertices that they contain; $s_i=|S_i|$. Given $S_1,\ldots,S_{i-1}$, the smallest numbered vertex $v_0$ not in $S_1\cup\cdots\cup S_{i-1}$ must be in $S_i$, and that must be the smallest numbered vertex in $S_i$ too. Let the cycle containing $v_0$ in $C$ be $(v_0,\ldots,v_{s_i-1})$ where $S_i=\{v_0,\ldots,v_{s_i-1}\}$. Then each $v_j$ must be one of the $k$ out-neighbors of $v_{j-1}$. Thus, given $S_1,\ldots,S_{i-1}$, the number of possibilities for $S_i$ is at most $k^{s_i-1}$. Therefore, the number of cycle covers of $G'$ associated with the tuple $(s_1,\ldots,s_r)$ is at most $k^{\sum_{i=1}^r(s_i-1)}=k^{n-r}$.

Finally, 
it is well known that the number of tuples $(s_1,\ldots,s_r)$, for a fixed $r$, such that $\sum_{i=1}^rs_i=n$ and each $s_i\geq2$, is at most ${{n-r+1}\choose{r-1}}<{{n}\choose{r}}$ for $r\leq n/2$. Thus, the number of cycle covers of $G'$ having $r$ cycles is at most ${{n}\choose{r}}k^{n-r}$.
\end{proof}

Now we are ready to prove Theorem \ref{thm_cyclecovercoloring}.

\begin{proof}[\textbf{Proof of Theorem \ref{thm_cyclecovercoloring}}]
Let $t=\gamma n/\ln k$, where $\gamma$ is a constant independent of $n$ as well as $k$, which we will fix later. Call a cycle cover \textit{bad} if it contains more than $t$ components; else call it \textit{good}. Call a cycle cover coloring $c:A\longrightarrow\{1,\ldots,k\}$ \textit{bad} if for each $i$, the cycle cover $c^{-1}(i)$ is bad; else call it \textit{good}. We need to prove an upper bound on the probability that the random cycle cover coloring sampled by the algorithm is bad. Note that if $k$ is small enough so that $\gamma/\ln k\geq1/2$ then every cycle cover is good, and hence, so is every cycle cover coloring. So assume $k$ is ``large enough'', that is, $\gamma/\ln k<1/2$, and thus $t<n/2$.

By Lemma \ref{lem_ub} the number of bad cycle covers is at most
\[\sum_{r=t+1}^{n/2}{{n}\choose{r}}k^{n-r}\leq\left(\frac{n}{2}-t\right){{n}\choose{t}}k^{n-t}\]
where the inequality follows from the fact that the function $r\longmapsto{{n}\choose{r}}k^{n-r}$ attains its maximum at $\lfloor\frac{n+1}{k+1}\rfloor<t$, and it is non-increasing in $\left[\lfloor\frac{n+1}{k+1}\rfloor,n\right]$. The number of bad cycle cover colorings is at most the number of ordered tuples of $k$ bad cycle covers, which is at most
\begin{eqnarray*}
\left(\frac{n}{2}-t\right)^k{{n}\choose{t}}^kk^{k(n-t)} & = & \left(\frac{n}{2}-\frac{\gamma n}{\ln k}\right)^k{{n}\choose{\frac{\gamma n}{\ln k}}}^kk^{k\left(n-\frac{\gamma n}{\ln k}\right)}\\
 & \leq & \left(\frac{n}{2}-\frac{\gamma n}{\ln k}\right)^k\left(\frac{e\ln k}{\gamma}\right)^{\frac{\gamma kn}{\ln k}}\left(\frac{k}{e^{\gamma}}\right)^{kn}
\end{eqnarray*}
Let $c$ be the random cycle cover coloring output by the algorithm. By Lemma \ref{lem_prob_ub} and equation (\ref{eqn_stirling}), the probability that $c$ is bad, is given by
\begin{eqnarray*}
\Pr[c\text{ is bad}] & \leq & \left(\frac{n}{2}-\frac{\gamma n}{\ln k}\right)^k\left(\frac{e\ln k}{\gamma}\right)^{\frac{\gamma kn}{\ln k}}\left(\frac{k}{e^{\gamma}}\right)^{kn}\times\left(\frac{e^2}{k}\right)^{kn}\times\frac{1}{2^{k-1}}\\
 & = & \left[\left(\frac{e\ln k}{\gamma}\right)^{\frac{\gamma}{\ln k}}\times\frac{1}{e^{\gamma-2}}\right]^{kn}\times\left(\frac{n}{2}-\frac{\gamma n}{\ln k}\right)^k\times\frac{1}{2^{k-1}}\\
 & = & q_k^{kn}\times\left(\frac{n}{2}-\frac{\gamma n}{\ln k}\right)^k\times\frac{1}{2^{k-1}}
\end{eqnarray*}
where
\[q_k=\left(\frac{e\ln k}{\gamma}\right)^{\frac{\gamma}{\ln k}}\times\frac{1}{e^{\gamma-2}}\]
We choose $\gamma$ so that for every ``large enough'' $k$, $q_k<1$. This ensures an exponential decay of failure probability of the algorithm with respect to $n$. We therefore need
\[\frac{1}{q_k}=\left(\frac{\gamma}{e\ln k}\right)^{\frac{\gamma}{\ln k}}\times e^{\gamma-2}>1\text{ i.e. }\frac{\gamma\ln\gamma}{\ln k}-\frac{\gamma(1+\ln\ln k)}{\ln k}+\gamma-2>0\]
It is sufficient to have, for every  ``large enough'' $k$,
\[\gamma\left[1-\frac{1+\ln\ln k}{\ln k}\right]>2\text{ i.e. }\gamma>\frac{2}{1-\frac{1+\ln\ln k}{\ln k}}=g(k)\text{ say.}\]
Thus, it is sufficient to ensure that for every $k$, either $\gamma\geq(\ln k)/2$ or $\gamma>g(k)$. That is, $\gamma>\max_k\min((\ln k)/2,g(k))$. It is easy to check that taking $\gamma=3.5$ suffices.
\end{proof}

\section{A Deterministic $\left(\frac{3}{2}+O\left(\frac{1}{\sqrt{k}}\right)\right)$-approximation Algorithm}\label{sec_det_alg}

The approach here is the same as the previous one: find a small number of cycles in the graph covering a large number of vertices, and connect them using a spanning tree. The main difference is that while we construct a cycle cover in the previous algorithm, here we find a collection of vertex-disjoint cycles covering almost half the vertices. As before, we contract the cycles, and connect them and the uncovered vertices together with a spanning tree. {\sc{LongCycles}}, given by Algorithm \ref{alg_longcycles}, essentially does a depth-first traversal, while repeatedly removing long cycles and vertices that cannot be fit in long cycles. From the description, it is clear that this algorithm runs in time $O(nk)$, and that it finds cycles of length no less than $d=2\sqrt{k}$. In order to derive the approximation ratio of our algorithm, we first need to bound from above the size of the set $B$ returned by {\sc{LongCycles}}. Let $m=|B|$.

\begin{algorithm}
\caption{{\sc{LongCycles}}($G$)}
\begin{algorithmic}[1]
\STATE \COMMENT{$G=(V,E)$: a $k$-regular $n$ vertex directed graph; returns a collection of cycles $\mathcal{C}$, each having length at least $2\sqrt{k}$, and a set $B$ of vertices not in any cycle in $\mathcal{C}$.}
\STATE Initialize $H:=G$, $\mathcal{C}=\emptyset$, $B:=\emptyset$, $P:=()$, $d=2\sqrt{k}$.
\STATE \COMMENT{$P$ always remains a path in $H$.}
\WHILE {$H$ is nonempty}
\IF {$P$ is empty}
\STATE Add an arbitrary vertex of $H$ to $P$.
\ELSE
\STATE \COMMENT{Suppose $P=(v_1,\ldots,v_t)$ with $t>0$.}
\IF {$v_t$ has a neighbor $u$ in $H$ outside $P$}
\STATE Append $u$ to $P$.
\ELSIF {$t\geq d$ and $v_t$ has a neighbor $v_s$ in $P$ for $s\leq t-d+1$}
\STATE Remove the vertices $v_s,v_{s+1},\ldots,v_{t-1},v_t$ from $P$ and from $H$; add this cycle to $\mathcal{C}$.
\ELSE
\STATE Remove $v_t$ from $P$ and from $H$, and add it to $B$.
\ENDIF
\ENDIF
\ENDWHILE
\STATE Return $\mathcal{C},B$.
\end{algorithmic}
\label{alg_longcycles}
\end{algorithm}

\begin{lemma}\label{lem_ub_b}
$m\leq\frac{n(k-2)}{2(k-d+1)}$
\end{lemma}

\begin{proof}
Suppose the set $B$ of vertices returned by the algorithm is $\{u_1,\ldots,u_m\}$, with the vertices added in the order $u_1,\ldots,u_m$. Consider the snapshot of the algorithm when the vertex $u_i$ was added to $B$. At that time, vertices $u_1,\ldots,u_{i-1}$ were already removed from $H$ and added to $B$, $u_{i+1},\ldots,u_m$ were still present in $H$, and $u_i$ was the last vertex in $P$. Let $u\in\{u_{i+1},\ldots,u_m\}$ a neighbor of $u_i$. Then $u$ must be in $P$, otherwise, some neighbor of $u_i$ would have been appended to $P$, rather than $u_i$ getting removed from $P$. Similarly, the distance between $u$ and $u_i$ on $P$ would be less than $d-1$, otherwise, a cycle would have been removed instead. Thus, the number of neighbors of $u_i$ among $u_{i+1},\ldots,u_m$ must be at most $d-2$. Therefore, the number of edges in the subgraph of $G$ induced by $B$ is less than $(d-2)m$. As a consequence, the number of edges in $G$ between $B$ and $V\setminus B$ is at least $km-2(d-2)m=(k-2d+4)m$.

Next, the number of vertices in $V\setminus B$ is $n-m$ and this is exactly the set of vertices covered by cycles in $\mathcal{C}$. For each vertex in $V\setminus B$, at most $k-2$ of the $k$ edges incident on it have their other endpoint in $B$. Thus, the number of edges in $G$ between $B$ and $V\setminus B$ is at most $(n-m)(k-2)$. Hence $(k-2d+4)m\leq(n-m)(k-2)$, which implies $m\leq\frac{n(k-2)}{2(k-d+1)}$.
\end{proof}

The above lemma implies that almost half of the vertices are covered by cycles in $\mathcal{C}$. We next use it to prove the approximation ratio.

\begin{theorem}
Consider the algorithm for finding a TSP tour, which runs {\sc{LongCycles}} on the input graph, and connects the cycles in $\mathcal{C}$ using two copies of a spanning tree of the graph obtained by contracting the cycles. The approximation ratio of this algorithm is $\frac{3}{2}+O\left(\frac{1}{\sqrt{k}}\right)$.
\end{theorem}

\begin{proof}
Since the vertex-disjoint cycles in $\mathcal{C}$ cover $n-m$ vertices, and each cycle contains at least $d$ vertices, the number of cycles in $\mathcal{C}$ is at most $(n-m)/d$, and hence, the number of components to be connected using a spanning tree is at most $(n-m)/d+m$. The TSP tour that the algorithm constructs consists of the cycles in $\mathcal{C}$, and two copies of a spanning tree in the graph obtained by contracting the cycles. The former contributes $n-m$ edges, while the latter contributes at most $2(n-m)/d+2m-2$ edges. Thus, the cost of the tour is at most
\begin{eqnarray*}
n-m+\frac{2(n-m)}{d}+2m-2 & = & n\left(1+\frac{2}{d}\right)+m\left(1-\frac{2}{d}\right)-2\\
 & \leq & n\left(1+\frac{2}{d}\right)+\frac{n(k-2)}{2(k-d+1)}\left(1-\frac{2}{d}\right)\\
 & \leq & n\left(1+\frac{2}{d}+\frac{k-2}{2(k-d+1)}\right)\\
 & = & n\left(\frac{3}{2}+\frac{2}{d}+\frac{d-3}{2(k-d+1)}\right)
\end{eqnarray*}
where we have used Lemma \ref{lem_ub_b} for the first inequality. For $d=\Theta\left(\sqrt{k}\right)$, the cost of the tour turns out to be $n\left(\frac{3}{2}+O\left(\frac{1}{\sqrt{k}}\right)\right)$. Thus, the algorithm achieves a $\frac{3}{2}+O\left(\frac{1}{\sqrt{k}}\right)$ factor approximation.
\end{proof}

\section{Concluding Remarks}


Both Vishnoi's algorithm as well as ours work only on regular graphs. Extending these to work on a larger class of graphs, with weaker assumptions about the vertex degrees, is an interesting problem, and will involve new techniques. We used the number of vertices as a lower bound on the cost of the optimal TSP tour. Extending to a larger class of graphs will require a tighter lower bound, and the cost of the Held-Karp relaxation is one candidate. Even for regular graphs, we do not know a hardness of approximation result, as a function of the degree $k$. Indeed improving the approximation factor to $1+c/k$ for some constant $c$ can not be ruled out. 

We would like to see whether our algorithm can be derandomized to get a $(1+o_k(1))$-approximation, possibly with some loss in the running time. We strongly feel that the following related avenues are worth exploring: first, to determine the best approximation ratio that can be achieved by deterministic algorithms for the {\sc{RegGraphTSP}}, and second, to determine the best approximation ratio that can be achieved by linear time deterministic algorithms. 

\section*{Acknowledgments}

The authors thank Nisheeth Vishnoi and Parikshit Gopalan for some initial discussions. The authors also thank Ayush Choure for his substantial contribution to this paper.

\bibliographystyle{plain}
\bibliography{references.bib}

\end{document}